\documentclass[conference]{IEEEtran}

\usepackage{enumerate}
\usepackage{tikz}
\usetikzlibrary{arrows,shapes,trees}
\usepackage{latexsym}
\usepackage{amssymb}
\usepackage[cmex10]{amsmath}
\usepackage{stmaryrd}
\usepackage{enumerate}

\usepackage{amsthm}

\newtheorem{theorem}{Theorem}

\newtheorem{lemma}[theorem]{Lemma}
\newtheorem{proposition}[theorem]{Proposition}
\newtheorem{example}[theorem]{Example}

\newcounter{myCounter}
\newtheorem{thmNo}[myCounter]{Theorem}

\newenvironment{theoremNo}[1]{\setcounter{myCounter}{0#1}\addtocounter{myCounter}{-1}\begin{thmNo}}{\end{thmNo}}

\newcommand{\until}[3]{#2 \mathrel{\mathbf{U}_{#3}} #1}
\newcommand{\since}[3]{#2 \mathrel{\mathbf{S}_{#3}} #1}
\newcommand{\fDia}{\Diamond}
\newcommand{\pDia}{\rlap{\hspace{0.24em}\raisebox{1pt}{-}}\Diamond}
\newcommand{\fBox}{\Box}
\newcommand{\pBox}{\rlap{\hspace{0.21em}\raisebox{1pt}{-}}\Box}
\newcommand{\true}{\mathbf{true}}
\newcommand{\false}{\mathbf{false}}
\newcommand{\reals}{\mathbb{R}_{\geq 0}}
\newcommand{\MLO}{$\mathit{FO}(<)$}
\newcommand{\ONEMLO}{$\mathit{FO}(<,+1)$}
\newcommand{\QMLO}{$\mathit{FO}(<,+\mathbb{Q})$}
\newcommand{\domain}{\mathbb{R}}

\newcommand{\maybeframe}[1]{\framebox{#1}}

\begin{document}
\title{Expressive Completeness for Metric Temporal Logic}

\author{
\IEEEauthorblockN{Paul Hunter, Jo\"el Ouaknine, and James Worrell}
\IEEEauthorblockA{Department of Computer Science\\University of Oxford\\
United Kingdom OX1 3QD\\
\texttt{\{paul.hunter,joel.ouaknine,james.worrell\}@cs.ox.ac.uk}}
}

\maketitle

\begin{abstract}
Metric Temporal Logic (MTL) is a generalisation of Linear Temporal
Logic in which the Until and Since modalities are annotated with
intervals that express metric constraints. A seminal result of Hirshfeld and
Rabinovich shows that over the reals, first-order logic with binary
order relation $<$ and unary function $+1$ is strictly more expressive
than MTL with integer constants.  Indeed they prove that no
temporal logic whose modalities are definable by formulas of bounded
quantifier depth can be expressively complete for $\mathit{FO}(<,+1)$.
In this paper we show the surprising result that if we allow unary
functions $+q$, $q \in \mathbb{Q}$, in first-order logic and
correspondingly allow rational constants in MTL, then the two logics
have the same expressive power.  This gives the first generalisation
of Kamp's theorem on the expressive completeness of LTL for
$\mathit{FO}(<)$ to the quantitative setting. The proof of this result
involves a generalisation of Gabbay's notion of separation.
\end{abstract}


%
\IEEEpeerreviewmaketitle

\section{Introduction}

One of the best-known and most widely studied logics in specification
and verification is \emph{Linear Temporal Logic (LTL)}: temporal logic
with the modalities \emph{Until} and \emph{Since}.  For discrete-time
systems one considers interpretations of \textit{LTL} over the integers
$(\mathbb{Z},<)$, and for continuous-time systems one considers
interpretations over the reals $(\mathbb{R},<)$.  A celebrated result
of Kamp~\cite{Kamp68} is that, over both $(\mathbb{Z},<)$ and
$(\mathbb{R},<)$, \textit{LTL} has the same expressiveness as the \emph{Monadic
  Logic of Order (\MLO)}: first-order logic with binary order relation
$<$ and uninterpreted monadic predicates.  Thus we can benefit from
the appealing variable-free syntax and elementary decision procedures
of \textit{LTL}, while retaining the expressiveness and canonicity of
first-order logic.

Over the reals {\MLO} cannot express metric properties, such as,
``every request is followed by a response within one time unit''.
This motivates the introduction of \emph{Monadic Logic of Order and
  Metric (\QMLO)}, which augments {\MLO} with a family of unary
function symbols $+q$, $q \in \mathbb{Q}$.  Correspondingly, there
have been a variety of proposals of quantitative temporal logics, with
modalities definable in {\QMLO}\@ (see,
e.g.,~\cite{AH-survey,AH93,AH94,Hen98,HRS98,HR04}).  Sometimes
attention is restricted to {\ONEMLO}---the fragment of {\QMLO} with
only the $+1$ function---and to temporal logics definable in this
fragment.  Typically these temporal logics can be seen as quantitative
extensions of \textit{LTL}\@.  However, until now there has been no fully
satisfactory counterpart to Kamp's theorem in the quantitative
setting.

The best-known quantitative temporal logic is \emph{Metric Temporal
  Logic (MTL)}, introduced over 20 years ago in~\cite{koymans}.  \textit{MTL}
arises by annotating the temporal modalities of \textit{LTL} with intervals
with rational endpoints, representing metric constraints.  Since the
\textit{MTL} operators are definable in {\QMLO}, it is immediate that
one can translate \textit{MTL} into {\QMLO}\@.  The main result of this paper
shows the converse, that \textit{MTL} is expressively complete for {\QMLO}\@.

The generality of allowing rational constants is crucial for
expressive completeness: our translation from {\QMLO} to \textit{MTL}
does not preserve the granularity of timing constraints.  Indeed, it
is known that \textit{MTL} with integer constants is not expressively
complete for {\ONEMLO}.  More generally, a seminal result of Hirshfeld
and Rabinovich~\cite[Theorem]{HR07} asserts that no temporal logic
whose modalities are definable by a (possibly infinite) set of
formulas of {\ONEMLO} of bounded quantifier depth can be
expressively complete for {\ONEMLO}.  Since the modalities of
\textit{MTL} are definable by formulas of quantifier depth two,
necessarily an \textit{MTL} formula equivalent to a given {\ONEMLO}
formula may require rational constants and itself only be definable in
{\QMLO}.

Two of the key ideas underlying our proof of expressive completeness
are \emph{boundedness} and \emph{separation}.  Given $N \in
\mathbb{N}$, an {\QMLO} formula $\varphi(x)$ is $N$-bounded if all
quantifiers are relativised to the interval $(x-N,x+N)$.  Exploiting a
normal form for {\MLO}, due to Gabbay, Pnueli, Shelah and
Stavi~\cite{GPSS80}, we show how to translate bounded {\QMLO} formulas
into \textit{MTL}\@.  Extending this translation to arbitrary {\QMLO}
formulas requires an appropriate metric analog of Gabbay's notion of
\emph{separation}~\cite{G81}.

Gabbay~\cite{G81} shows that every \textit{LTL} formula can be equivalently
rewritten as a Boolean combination of formulas, each of which depends
only on the past, present or future.  This seemingly innocuous
separation property has several far-reaching consequences (see the
survey of Hodkinson and Reynolds~\cite{HodkinsonR05}).  In particular,
the fact \textit{LTL} has the property is a key lemma in an inductive
translation from {\MLO} to \textit{LTL}\@.  We prove an analogous result for
\textit{MTL}: every \textit{MTL} formula can be equivalently rewritten as a Boolean
combination of formulas, each of which is either bounded (i.e., refers
to the \emph{near present}) or refers to the \emph{distant future} or
\emph{distant past}.  Crucially, while the distant past and distant
future are disjoint, they are both allowed to overlap with near
present, unlike in Gabbay's result.  We exploit our result in like
manner to Gabbay to give an inductive translation of {\QMLO} to \textit{MTL}\@.
Here it is vital that we already have a translation of bounded {\QMLO}
formulas to \textit{MTL}.

\subsection*{Related Work}

A more elaborate quantitative extension of \textit{LTL} is \emph{Timed
  Propositional Temporal Logic (TPTL)}, which expresses timing
constraints using variables and
{freeze quantification}~\cite{AH94}.  From the respective
definitions of the logics the following inclusions in
expressiveness are straightforward:
\begin{gather} \mathit{MTL} \subseteq \mathit{TPTL} 
\subseteq \mbox{{\QMLO}} \, . 
\label{eq:chain}
\end{gather}
Bouyer, Chevalier and Markey~\cite{BCM05} showed that the inclusion
between \textit{MTL} and \textit{TPTL} is strict if only future temporal connectives are
considered, confirming a conjecture of~\cite{AH94}.  However they left
open the case in which both past and future connectives are allowed.
Our main result shows that in this case the chain of inclusions
(\ref{eq:chain}) collapses, resolving this open question.

In fact, \textit{TPTL} has already been shown to be expressively complete for
{\QMLO} in~\cite{HV07}.  Notwithstanding this
result, we regard the result in the present paper as the first fully
satisfactory analog of Kamp's Theorem for {\QMLO}\@.  This is because
\textit{TPTL} is a hybrid between first-order logic and temporal logic,
featuring variables and quantification in addition to temporal
modalities~\cite{H90}.

The logics considered in this paper are all undecidable.  Adding $+1$
to {\MLO} or $\fDia_{=1}\varphi$ ($\varphi$ will be true in exactly
one time unit) to \textit{LTL} already leads to an undecidable
satisfiability problem over the reals.  Intuitively, the source of
undecidability is the ability to express \emph{punctual} metric
constraints, such as ``every request is followed by a response in
exactly one time unit''.  The expressiveness of decidable quantitative
temporal logics has also been investigated in~\cite{HR05,HR07}.  The
main results present a hierarchy of decidable temporal logics with
\emph{counting modalities}, and characterise their expressiveness in
terms of fragments of $\mathit{FO}(<,+1)$.  The most basic such logic
arises by adding the modality $\fDia_{<1} \varphi$ to \textit{LTL},
expressing that $\varphi$ will be true within one time unit in the
future, and is equivalent in expressiveness to the logic
MITL~\cite{AFH96}.
 

Yet another approach to expressive completeness is taken in our
previous work~\cite{ORW09}.  This paper considers the fragment of
{\QMLO} with only the $+1$ function.  Likewise it restricts to \textit{MTL}
formulas in which intervals have integer endpoints.  Recall that in
this setting expressive completeness fails over unbounded domains such
as $(\mathbb{R},<)$ and $(\mathbb{R}_{\geq 0},<)$.
However~\cite{ORW09} shows that expressive completeness holds over
each bounded time domain $([0,N),<)$.  While some of the ideas
  from~\cite{ORW09} are used in the present paper, our results differ
  substantially.  While~\cite{ORW09} relies on a bounded time domain,
  the present paper considers syntactically bounded formulas on an
  unbounded domain.  Even the fact that \textit{MTL} is expressively complete
  for syntactically bounded {\ONEMLO} formulas crucially uses the fact
  that we allow fractional constants.

\section{Definitions and Main Results}\label{sec:prelim}

\subsection{First-order logic}
Formulas of \emph{Monadic Logic of Order and Metric ({\QMLO})} are
first-order formulas over a signature with a binary relation symbol
$<$, an infinite collection of unary predicate symbols
$P_1,P_2,\ldots$, and an infinite family of unary function symbols
$+q$, $q \in \mathbb{Q}$.  Formally, the terms of {\QMLO} are
generated by the grammar $t::=x\mid t+q$, where $x$ is a variable and
$q \in \mathbb{Q}$.  Formulas of {\QMLO} are given by the
following syntax:
\[ \varphi ::= \true \mid P_i(t) \mid t < t \mid \varphi \wedge
\varphi \mid \neg \varphi \mid \exists x \, \varphi \, , \] where $x$
denotes a variable and $t$ a term.  

We consider interpretations of {\QMLO} over the real line\footnote{
  Our results carry over to subintervals of $\mathbb{R}$, such as the
  non-negative reals $\reals$.}, $\mathbb{R}$, with the natural
interpretations of $<$ and $+q$.  It follows that a structure for
{\QMLO} is determined by an interpretation of the monadic predicates.

Of particular importance is {\ONEMLO}, the fragment of {\QMLO} that
omits all the $+q$ functions except $+1$.  For simplicity, when
considering formulas of {\ONEMLO} we will often use standard
arithmetical notation as a shorthand, for example,
\[ x-y>2 \quad \equiv \quad (y+1)+1 < x \, . \]

\subsection{Metric Temporal Logic}
Given a set $\boldsymbol{P}$ of atomic propositions, the formulas of
\emph{Metric Temporal Logic (MTL)} are built from $\boldsymbol{P}$
using Boolean connectives and time-constrained versions of the
\emph{Until} and \emph{Since} operators $\until{}{}{}$ and
$\since{}{}{}$ as follows:
\[ \varphi ::= \true\mid P \mid \varphi \wedge \varphi \mid \neg \varphi \mid
\until{\varphi}{\varphi}{I} \mid \since{\varphi}{\varphi}{I} \, , \]
where $P \in \boldsymbol{P}$ and $I\subseteq (0,\infty)$ is an
interval with endpoints in $\mathbb{Q}_{\geq 0}\cup\{\infty\}$.

Intuitively, the meaning of $\until{\varphi_2}{\varphi_1}{I}$ is that
$\varphi_2$ will hold at some time in the interval $I$, and until then
$\varphi_1$ holds. 
More precisely, the semantics of \textit{MTL} are defined as follows.  A
\emph{signal} is a function $f:\domain \to 2^{\boldsymbol{P}}$.  Given a signal $f$
and $r \in \domain$, we define the satisfaction relation $f,r \models
\varphi$ by induction over $\varphi$ as follows:
\begin{itemize}
\item $f,r \models p$ iff $p \in f(r)$,
\item $f,r \models \neg \varphi$ iff $f,r \not\models \varphi$,
\item $f,r \models \varphi_1 \wedge \varphi_2$ iff $f,r \models \varphi_1$ and $f,r \models \varphi_2$,
\item $f,r \models \until{\varphi_2}{\varphi_1}{I}$ iff there exists
  $t>r$ such that $t-r \in I$, $f,t \models \varphi_2$ and $f,u \models
  \varphi_1$ for all $u$, $r < u < t$,
\item $f,r \models \since{\varphi_2}{\varphi_1}{I}$ iff there exists
  $t<r$ such that $r-t \in I$, $f,t \models \varphi_2$ and $f,u
  \models \varphi_1$ for all $u$, $t < u < r$.
\end{itemize}

\textit{LTL} can be seen as a restriction of \textit{MTL} with only the interval
$I=(0,\infty)$.  Indeed, if $I=(0,\infty)$ then we omit the annotation
$I$ in the corresponding temporal operator since the constraint is
vacuous.  We also use arithmetic expressions to denote
intervals.  For example, we write $\until{}{}{<3}$ for
$\until{}{}{(0,3)}$ and $\until{}{}{=1}$ for $\until{}{}{\{1\}}$.  We
say the $\until{}{}{I}$ and $\since{}{}{I}$ operators are
\emph{bounded} if $I$ is bounded, otherwise we say that the operators
are \emph{unbounded}.

We introduce the derived connectives $\fDia_I \varphi:=
\until{\varphi}{\true}{I}$ ($\varphi$ will be true at some point in
interval $I$) and $\pDia_I \varphi :=\since{\varphi}{\true}{I}$
($\varphi$ was true at some point in interval $I$ in the past).  We
also have the dual connectives $\fBox_I \varphi := \neg \fDia_I \neg
\varphi$ ($\varphi$ will hold at all times in interval $I$ in the
future) and $\pBox_I := \neg \pDia_I \neg \varphi$ ($\varphi$ was true
at all times in interval $I$ in the past).

\subsection{Expressive Equivalence}
\label{sec:exp-equiv}
Given a set $\boldsymbol{P}=\{P_1, \ldots, P_m\}$ of monadic
predicates, a signal $f:\domain \to 2^{\boldsymbol{P}}$ defines an
interpretation of each $P_i$, where $P_i(r)$ if and only if $P_i \in
f(r)$.  As observed earlier, this is sufficient to define the
model-theoretic semantics of {\QMLO}, enabling us to relate the
semantics of {\QMLO} and \textit{MTL}.

Let $\varphi(x)$ be an {\QMLO} formula with one free variable and $\psi$
an \textit{MTL} formula.  We say $\varphi$ and $\psi$ are \emph{equivalent} if
for all signals $f$ and $r \in \domain$:
\[ f \models \varphi[r] \Longleftrightarrow f,r \models \psi.\]

\begin{example}
\label{ex:fomtl}
Consider the following formula, which says that $P$ will be true at
two points within the next time unit:
\[ 
\varphi(x) := \exists y\, \exists z\, ((x<y<z<x+1) \wedge P(y) \wedge
P(z)) \, . \] It was shown in~\cite{HR07} that $\varphi$ cannot be
expressed in \textit{MTL} using only integer constants\footnote{In
  fact~\cite{HR07} did not consider so-called \emph{punctual}
  operators, i.e., singleton constraining intervals.  But their
  argument goes through \emph{mutatis mutandis}.}.  To see this,
consider the signal $f$ in which the predicate $P$ is true exactly at
the points $\frac{2n}{3}$, $n \in \mathbb{N}$.  It can be shown by
induction that for every \textit{MTL} formula $\varphi$ with integer constants
there exists $t_0>0$ and a predicate $\theta$ that is either
$\mathbf{true}$, $\mathbf{false}$, $P$, $\neg P$, or $\fDia_{=1}P$,
such that for all $t>t_0$, $f,r \models \varphi$ iff $f,r\models
\theta$.  On the other hand, for $2n\equiv 1 \;(\mathrm{mod}\; 3)$, $\varphi$ is
continuously true on the interval $(\frac{2n-1}{3},\frac{2n}{3})$ and
false on the boundary of the interval.

As observed in~\cite{BCM05}, we can, however, express $\varphi(x)$ in \textit{MTL}
by using fractional constants.  The idea is to consider three cases
according to whether $P$ is true twice in the interval
$(x,x+\frac{1}{2}]$, twice in the interval $[x+\frac{1}{2},x+1)$, or 
once each in $(x,x+\frac{1}{2})$ and $(x+\frac{1}{2},x+1)$.  We are thus led to define
the \textit{MTL} formula
\begin{align*}
\varphi^\dag \, := \, & \fDia_{(0,\frac{1}{2})}(P \wedge \fDia_{(0,\frac{1}{2})} P) \, \vee\\ 
				& \fDia_{=1}(\pDia_{(0,\frac{1}{2})}(P \wedge \pDia_{(0,\frac{1}{2})}P)) \,\vee\\
                &(\fDia_{(0,\frac{1}{2})}P \wedge \fDia_{(\frac{1}{2},1)}P)\,,
\end{align*}
which is equivalent to $\varphi$.
\end{example}

The following is straightforward.
\begin{proposition}\label{prop:easy}
For every \textit{MTL} formula $\varphi$ there is an equivalent {\QMLO} formula
$\varphi^\ast(x)$.
\end{proposition}

Our main result is the converse:
\begin{theorem}
For every {\QMLO} formula $\varphi(x)$ there is an equivalent {MTL} formula
$\varphi^\dag$.
\label{thm:main}
\end{theorem}

As we now explain, by a simple scaling argument it suffices to prove
Theorem~\ref{thm:main} in the special case for which $\varphi$ is an
{\ONEMLO} formula.  Let $f$ be a signal and $r \in
\mathbb{Q}_{>0}$.  We define the signal $r\cdot f$ by $r \cdot f(s) := f(\frac{s}{r})$.
Given either an {\QMLO} formula $\varphi(x)$ or an \textit{MTL} formula $\varphi$,
we say that the formula $\varphi^r$ is a \emph{scale} of $\varphi$ by
$r\in \mathbb{Q}_{>0}$, if for all signals $f$ and all $s \in \domain$,
\[ f,s \models \varphi \quad \Longleftrightarrow 
\quad r \cdot f, rs \models \varphi^r \, . \] 
It is straightforward that
{\QMLO} and \textit{MTL} are both closed under scaling: in each case the
required formula $\varphi^r$ is obtained by multiplying all constants
occurring in $\varphi$ by $r$.

Now we show how to deduce expressive completeness of \textit{MTL} for {\QMLO} from
the fact that \textit{MTL} is at least as expressive as the fragment
{\ONEMLO}.  Given an {\QMLO} formula $\varphi(x)$, pick $r$ such that
$\varphi^r$ is an {\ONEMLO} formula and translate $\varphi^r$ to an
equivalent \textit{MTL} formula $\psi$.  Then rescaling $\psi$ by $1/r$, we
obtain an \textit{MTL} formula $\psi^{1/r}$ that is equivalent to the original
formula $\varphi$.

We will see later that the translation from {\ONEMLO} to \textit{MTL} already
involves temporal operators whose constraining intervals have
fractional endpoints, as suggested by Example~\ref{ex:fomtl}.

\section{Syntactic Separation of MTL}\label{sec:sep}

In~\cite{GHR94}, Gabbay \textit{et al.} showed that \textit{LTL} formulas over
Dedekind-complete domains are equivalent to Boolean combinations of
formulas that depend exclusively on one of the past, present, or
future.  We state this result as it applies to continuous domains (the
formulation in the discrete setting is slightly more straightforward).
To state the result we recall the \emph{right-limit} modality $K^+$ and 
\emph{left-limit} modality $K^-$,
respectively defined as:
\[ K^+ \varphi  := \neg (\until{\true}{\neg \varphi}{}) \qquad K^- \varphi  := \neg (\since{\true}{\neg \varphi}{})\, .\]
The formula $K^+ \varphi$ states that $\varphi$ is true 
arbitrarily close in the future and $K^- \varphi$ asserts that $\varphi$ is
true arbitrarily close in the past.

\begin{theorem}[\cite{GHR94}]\label{thm:gabbaySep}
Over Dedekind-complete domains, every {LTL} formula is equivalent to a
Boolean combination of:
\begin{itemize}
\item atomic formulas,
\item formulas of the form $\until{\varphi_2}{\varphi_1}{}$ such that
  $\varphi_1$ and $\varphi_2$ use only $\until{}{}{}$ and $K^-$, 
\item formulas of the form $\since{\varphi_2}{\varphi_1}{}$ such that
$\varphi_1$ and $\varphi_2$ use only $\since{}{}{}$ and $K^+$.
\end{itemize}
\end{theorem}
Note that the three classes of formulas in Theorem~\ref{thm:gabbaySep}
respectively refer to the present, future and past.  In this section
we derive an analogous result for \textit{MTL}\@.  We show that every \textit{MTL}
formula can be written as a Boolean combination of \emph{bounded},
\emph{distant future} and \emph{distant past formulas}.  Just as
Gabbay \textit{et al.} used syntactic forms for future and past
representations, our plan is to use natural forms for bounded, distant
future and distant past formulas.  
Crucially, the distant future and
distant past are allowed to overlap with the bounded present, unlike
in the result of Gabbay \textit{et al}.

Given an \textit{MTL} formula $\varphi$, we define the \emph{future-reach}
$\mathit{fr}(\varphi)$ and \emph{past-reach} $\mathit{pr}(\varphi)$
inductively as follows:
\begin{itemize}
\item $\mathit{fr}(p) = \mathit{pr}(p)=0$ for all propositions $p$,
\item $\mathit{fr}(\true) = \mathit{pr}(\true)= 0$,
\item $\mathit{fr}(\neg \varphi) = \mathit{fr}(\varphi)$, $\mathit{pr}(\neg \varphi) = \mathit{pr}(\varphi)$,
\item $\mathit{fr}( \varphi \wedge \psi) = \max\{\mathit{fr}(\varphi),\mathit{fr}(\psi)\}$,
\item $\mathit{pr}( \varphi \wedge \psi) = \max\{\mathit{pr}(\varphi),\mathit{pr}(\psi)\}$,
\item If $n=\inf(I)$ and $m=\sup(I)$:
\begin{itemize}
\item $\mathit{fr}(\until{\psi}{\varphi}{I}) = m + \max\{\mathit{fr}(\varphi),\mathit{fr}(\psi)\}$,
\item $\mathit{pr}(\since{\psi}{\varphi}{I}) = m + \max\{\mathit{pr}(\varphi),\mathit{pr}(\psi)\}$,
\item $\mathit{fr}(\since{\psi}{\varphi}{I}) = \max\{\mathit{fr}(\varphi),\mathit{fr}(\psi)-n\}$,
\item $\mathit{pr}(\until{\psi}{\varphi}{I}) = \max\{\mathit{pr}(\varphi),\mathit{pr}(\psi)-n\}$.
\end{itemize}
\end{itemize}
Intuitively the future-reach indicates how much of the future is required to determine the truth of an \textit{MTL} formula, and likewise for the past-reach.
%
%
Note that if $\varphi$ contains an unbounded $\until{}{}{}$ operator
then $\mathit{fr}(\varphi) = \infty$ and likewise if $\varphi$ contains an
unbounded $\since{}{}{}$ operator, $\mathit{pr}(\varphi) = \infty$.

We say an \textit{MTL} formula is \emph{syntactically separated} if it is a
Boolean combination of the following
\begin{itemize}
\item $\fDia_{=N}  \varphi$ where $\mathit{pr}(\varphi)<N-1$,
\item $\pDia_{=N} \varphi$ where $\mathit{fr}(\varphi)< N-1$, 
\item $\varphi$, where all intervals occurring in temporal operators
  are bounded.
\end{itemize}
We call formulas of the third kind above \emph{bounded}.  Note that
formulas with no occurrences of $\until{}{}{I}$ and $\since{}{}{I}$ are
included in the definition of bounded formulas.

\begin{example}
Consider the formula $\varphi= \fDia \pBox (p \rightarrow
\pDia_{=1}p)$.  Then
$\mathit{fr}(\varphi)=\mathit{pr}(\varphi)=\infty$.  We define an equivalent separated formula as follows.  First, write $\psi = p \rightarrow \pDia_{=1}p$.
Then $\varphi$ is equivalent to
\begin{align*}
& \pDia_{=1}(\psi \wedge \pBox \psi) \wedge \pBox_{(0,1)} \psi \wedge \psi\\ 
& \wedge
 \big( (\until{\psi}{\psi}{\leq 2}) 
\vee
         (\fBox_{\leq 2}\psi \wedge \fDia_{=2}(\until{\psi}{\psi}{})) \big) \, .
\end{align*}
\end{example}

\begin{theorem}\label{thm:sepThm}
Every MTL formula is equivalent to one which is syntactically separated.
\end{theorem}

\noindent To prove Theorem~\ref{thm:sepThm} our strategy is as follows:
\begin{enumerate}[Step 1. ]
\item Remove all unbounded $\until{}{}{}$ and $\since{}{}{}$ operators
  from within the scope of bounded operators.
\item Treating bounded formulas as atoms, apply
  Theorem~\ref{thm:gabbaySep} to remove unbounded $\until{}{}{}$
  operators from the scope of unbounded $\since{}{}{}$ operators and
  vice versa.  
\item Divide the top-level unbounded operators into formulas
  bounded by $N$ and formulas at least $N$ away for sufficiently large
  $N$ to separate these formulas.  This step may also place unbounded
  operators within the scope of bounded operators, but still maintains
  the separation of unbounded $\until{}{}{}$ and unbounded
  $\since{}{}{}$ operators.  Using Step 1, and observing that this
  does not introduce any new unbounded operators, we can move these
  unbounded operators to the top level and recursively apply the
  division to completely separate the formula.
\end{enumerate}

\subsubsection*{Step 0.  Translation to Normal Form}
We first introduce a normal form for \textit{MTL} formulas.  In
defining this we regard $\until{}{}{I}$, $\since{}{}{I}$, $\fBox_I$,
$\fDia_I$, $\pBox_I$, and $\pDia_I$ as primitive operators.  Then an
\textit{MTL} formula is said to be in \emph{normal form} if the
following all hold:
\begin{enumerate}[(i)]
\item The formula is written using the Boolean operators and the
  temporal connectives $\until{}{}{(0,\gamma)}$, $\since{}{}{(0,\gamma)}$,
  $\fBox_{(0,\gamma)}$, $\pBox_{(0,\gamma)}$, where $\gamma \in
  \mathbb{Q}_{\geq 0}\cup \{\infty\}$, and $\fDia_{=q}$ and $\pDia_{=q}$, where
  $q\in \mathbb{Q}_{\geq 0}$;
\item In any subformula $\until{\varphi_2}{\varphi_1}{I}$ or
  $\since{\varphi_2}{\varphi_1}{I}$, the outermost connective of
  $\varphi_1$ is not conjunction and the outermost connective of
  $\varphi_2$ is not disjunction;
\item No temporal operator occurs in the scope of $\fDia_{=q}$ or
  $\pDia_{=q}$;
\item Negation is only applied to propositional variables and bounded temporal operators.
\end{enumerate}

We can transform an \textit{MTL} formula into an equivalent normal form as
follows.  To satisfy (i) we eliminate connectives $\until{}{}{I}$
and $\since{}{}{I}$ in which the interval $I$ does not have left
endpoint $0$ using the equivalences
\begin{eqnarray*}
\until{\psi}{\varphi}{(p,q)} & \longleftrightarrow & \fBox_{(0,p)} \varphi \wedge \fDia_{=p} 
\big(\varphi \wedge (\until{\psi}{\varphi}{(0,q-p)})\big)\\
\since{\psi}{\varphi}{(p,q)} & \longleftrightarrow & \pBox_{(0,p)} \varphi \wedge \pDia_{=p} 
\big(\varphi \wedge (\since{\psi}{\varphi}{(0,q-p)})\big)
\end{eqnarray*}
and corresponding equivalences for left-closed and right-closed
intervals.

To satisfy (ii) we use the equivalences
\begin{eqnarray*}
\until{(\psi \vee \theta)}{\varphi}{I} & \longleftrightarrow &
(\until{\psi}{\varphi}{I}) \vee
(\until{\theta}{\varphi}{I})\\ \until{\theta}{(\varphi \wedge
  \psi)}{I} & \longleftrightarrow &
(\until{\theta}{\varphi}{I}) \wedge (\until{\theta}{\psi}{I})
\end{eqnarray*}
and their corresponding versions for $\since{}{}{I}$,
\begin{eqnarray*}
\since{(\psi \vee \theta)}{\varphi}{I} & \longleftrightarrow &
(\since{\psi}{\varphi}{I}) \vee
(\since{\theta}{\varphi}{I})\\ \since{\theta}{(\varphi \wedge
  \psi)}{I} & \longleftrightarrow &
(\since{\theta}{\varphi}{I}) \wedge (\since{\theta}{\psi}{I}) \, .
\end{eqnarray*}

To satisfy (iii) we use the equivalences
\begin{eqnarray*}
\fDia_{=q}(\varphi \wedge \psi) & \longleftrightarrow & \fDia_{=q}\varphi \wedge \fDia_{=q}\psi\\
\fDia_{=q}(\neg \varphi) & \longleftrightarrow & \neg \fDia_{=q} \varphi\\
\fDia_{=q} (\until{\psi}{\varphi}{I}) & \longleftrightarrow & \until{\fDia_{=q}\psi}{\fDia_{=q}\varphi}{I}\\
\fDia_{=q} (\since{\psi}{\varphi}{I}) & \longleftrightarrow & \since{\fDia_{=q}\psi}{\fDia_{=q}\varphi}{I}
\end{eqnarray*}
and the corresponding equivalences for $\pDia_{=q}$ to distribute
$\fDia_{=q}$ and $\pDia_{=q}$ across all other operators.

To satisfy (iv) we observe that the $K^+$ and $K^-$ operators can be
defined as bounded formulas, \textit{viz.}
\[ K^+(\varphi) \: \leftrightarrow\: \neg (\until{\true}{\neg \varphi}{<1})\qquad K^-(\varphi) \: \leftrightarrow\: \neg (\since{\true}{\neg \varphi}{<1})
\, .\]  Then we 
use the equivalences
\begin{eqnarray*}
\neg(\until{\psi}{\varphi}{}) & \longleftrightarrow &
\fBox \neg\psi \vee K^+(\neg \varphi) \vee\\
&&\quad(\until{(\neg\psi \wedge (\neg\varphi\vee K^+(\neg\varphi)))}{\neg\psi}{}) \\
\neg \fBox \varphi & \longleftrightarrow &
\until{\neg\varphi}{\true}{} 
\end{eqnarray*}
and their corresponding past versions to rewrite any subformula in
which negation is applied to an unbounded temporal operator.

\subsection*{Step 1. Extracting unbounded Until and Since}
Our goal in this subsection is the following lemma.
\begin{lemma}\label{lem:extract}
Every {MTL} formula $\varphi$ is equivalent to one in which no unbounded
temporal operator occurs within the scope of a bounded temporal
operator.
\end{lemma}

The proof of this lemma relies on Proposition~\ref{prop:oneStepElim},
whose proof is straightforward.

\begin{proposition}\label{prop:oneStepElim}
For all $q \in \mathbb{Q}_{\geq 0}$, the following equivalences and their temporal duals hold over all signals.
\begin{enumerate}[(i) ]
\addtolength{\itemsep}{3ex}
\item \maybeframe{\parbox{0.4\textwidth}{\[\begin{array}[t]{c}\until{\bigl((\until{\psi}{\varphi}{})\wedge \chi\bigr)}{\theta}{<q}\\[1.5ex]\leftrightarrow\\[1.5ex]\until{\bigl((\until{\psi}{\varphi}{<q})\wedge \chi\bigr)}{\theta}{<q}\quad\vee\\\bigg(\bigl(\until{(\fBox_{<q} \varphi \wedge \chi)}{\theta}{<q}\bigr) \: \wedge\: \fDia_{=q} (\until{\psi}{\varphi}{})\bigg)\end{array}\]}}

\item \maybeframe{\parbox{0.4\textwidth}{\[\begin{array}[t]{c}\until{(\fBox \varphi \wedge \chi)}{\theta}{<q} \\[1.5ex]\leftrightarrow \\[1.5ex] \bigl(\until{(\fBox_{<q} \varphi \wedge \chi)}{\theta}{<q}\bigr) \wedge \fDia_{=q} \fBox \varphi\end{array}\]}}
\item \maybeframe{\parbox{0.4\textwidth}{\[\begin{array}[t]{c}\until{\bigl((\since{\psi}{\varphi}{})\wedge \chi\bigr)}{\theta}{<q} \\[1.5ex] \leftrightarrow \\[1.5ex]  
\until{\bigl((\since{\psi}{\varphi}{<q})\wedge \chi\bigr)}{\theta}{<q}\quad\vee\\\bigg(\bigl(\until{(\pBox_{<q} \varphi \wedge \chi)}{\theta}{<q}\bigr) \: \wedge\: \since{\psi}{\varphi}{}\bigg)\end{array}\]}}
\item \maybeframe{\parbox{0.4\textwidth}{\[\begin{array}[t]{c}\until{(\pBox \varphi \wedge \chi)}{\theta}{<q} \\[1.5ex] \leftrightarrow \\[1.5ex] \bigl(\until{(\pBox_{<q} \varphi \wedge \chi)}{\theta}{<q}\bigr) \wedge \pBox \varphi\end{array}\]}}
\item \maybeframe{\parbox{0.4\textwidth}{\[\begin{array}[t]{c} \until{\theta}{\bigl((\until{\psi}{\varphi}{})\vee \chi\bigr)}{<q} \\[1.5ex] \leftrightarrow \\[1.5ex]
\until{\theta}{\bigl((\until{\psi}{\varphi}{<q})\vee \chi\bigr)}{<q}\quad \vee\\\Big[\bigl(\until{(\fBox_{<q}\varphi)}{(\until{\psi}{\varphi}{<q}) \vee \chi \bigr)}{<q}\\\:\wedge\: \fDia_{<q}\theta \: \wedge\: \fDia_{= q}(\until{\psi}{\varphi}{})\Big]\end{array}\]}}
\item \maybeframe{\parbox{0.4\textwidth}{\[\begin{array}[t]{c} \until{\theta}{\bigl((\fBox \varphi)\vee \chi\bigr)}{<q} \\[1.5ex] \leftrightarrow \\[1.5ex]
\until{\theta}{\chi}{<q}\quad\vee\\\bigl(\until{(\fBox_{<q}\varphi)}{\chi}{<q}\:\wedge\: \fDia_{<q}\theta \: \wedge\: \fDia_{= q}(\fBox \varphi)\bigr)\end{array}\]}}
\item \maybeframe{\parbox{0.4\textwidth}{\[\begin{array}[t]{c} \until{\theta}{\bigl((\since{\psi}{\varphi}{})\vee \chi\bigr)}{<q} \\[1.5ex] \leftrightarrow \\[1.5ex]
\until{\theta}{\bigl((\since{\psi}{\varphi}{<q})\vee \chi\bigr)}{<q}\quad\vee\\ \Big[\until{\theta}{\bigl(\pBox_{<q}\varphi \vee (\since{\psi}{\varphi}{<q}) \vee \chi\bigr)}{<q}\:\wedge\: (\since{\psi}{\varphi}{})\Big] \end{array}\]}}
\item \maybeframe{\parbox{0.4\textwidth}{\[\begin{array}[t]{c}\until{\theta}{\bigl(\pBox \varphi\vee \chi\bigr)}{<q} \\[1.5ex] \leftrightarrow \\[1.5ex]
\until{\theta}{\chi}{<q}\:\vee\\\Big[\bigl(\until{\theta}{(\pBox_{<q}\varphi \vee \chi)}{<q}\bigr)\: \wedge\: \pBox \varphi\Big].\end{array}\]}}
\end{enumerate}
\end{proposition}

\subsubsection*{Proof of Lemma~\ref{lem:extract}}
Define the \emph{unbounding depth} $\mathit{ud}(\varphi)$ of an \textit{MTL}
formula $\varphi$ to be the modal depth of $\varphi$, counting only
unbounded temporal operators.  Thus we have
\[ \mathit{ud}(\until{\varphi_2}{\varphi_1}{I}) =
   \left\{ \begin{array}{ll}
   \max(\mathit{ud}(\varphi_1),\mathit{ud}(\varphi_2)) &\mbox{$I$ bounded}\\
   \max(\mathit{ud}(\varphi_1),\mathit{ud}(\varphi_2))+1 & \mbox{otherwise}
   \end{array}\right . \]
with similar clauses for the other temporal operators.

Now suppose that $\varphi$ is an \textit{MTL} formula in normal form in which
some unbounded temporal operator occurs within the scope of a bounded
temporal operator.  Then some subformula of $\varphi$ (or its temporal dual) matches the
top side of one of the equivalences in
Proposition~\ref{prop:oneStepElim}.  Pick such a subformula $\psi$
with maximum unbounding depth $\mathit{ud}(\psi)$ and replace it with
the bottom side $\psi'$ of the corresponding equivalence.  Notice
that all subformulas of $\psi'$ whose outermost connective is a
bounded temporal operator other than $\fDia_{=q}$ and $\pDia_{=q}$
have unbounding depth strictly less than $\mathit{ud}(\psi)$.  Finally
rewrite $\psi'$ to normal form, in particular pushing the newly
introduced $\fDia_{=q}$ and $\pDia_{=q}$ operators inward.  Notice
that this last step does not increase the maximum unbounding depth.

This rewriting process must eventually terminate, yielding a formula
in which no unbounded operator remains within the scope of a bounded
operator.

\subsection*{Step 2. Extracting Since from Until and vice-versa}
Now suppose we have an \textit{MTL} formula in which no unbounded temporal
operator occurs within the scope of a bounded operator.  If we replace
each bounded subformula $\theta$ with a new proposition $P_\theta$, the
resulting formula is now an \textit{LTL} formula equivalent to our original
formula for suitable interpretations of the $P_\theta$.  From
Theorem~\ref{thm:gabbaySep} we know that this formula is equivalent to
a Boolean combination of:
\begin{itemize}
\item  atomic formulas,
\item formulas of the form
$\until{\varphi_1}{\varphi_2}{}$ such that $\varphi_1$ and $\varphi_2$
use only $\until{}{}{}$ and $K^-$,
\item formulas of the form
$\since{\varphi_1}{\varphi_2}{}$ such that $\varphi_1$ and $\varphi_2$
use only $\since{}{}{}$ and $K^+$.
\end{itemize}

Recalling from Step 0 that we can express the operators $K^+$ and
$K^-$ using bounded operators, and also replacing each proposition
$P_\theta$ with its associated bounded formula $\theta$, we obtain:
\begin{lemma}\label{lem:separate}
Every {MTL} formula is equivalent to a Boolean combination of: 
\begin{itemize}
\item bounded formulas,
\item formulas that use arbitrary $\until{}{}{I}$ but only bounded
  $\since{}{}{I}$,
\item formulas that use arbitrary $\since{}{}{I}$ but only bounded
  $\until{}{}{I}$
\end{itemize}
\end{lemma}

\subsection*{Step 3. Completing the separation}
Now suppose we have an \textit{MTL} formula $\theta$ that does not contain
unbounded $\since{}{}{}$.  We prove by induction on the number of
unbounded $\until{}{}{}$ operators that $\theta$ is equivalent to a
syntactically separated formula.  Clearly if $\theta$ contains no
unbounded $\until{}{}{}$ operators then it is bounded and therefore
syntactically separated.  Otherwise, by applying
Lemma~\ref{lem:extract} and observing that it does not introduce
unbounded $\until{}{}{}$ operators, we may assume that $\theta =
\until{\psi}{\varphi}{}$ where $\varphi$ and $\psi$ have strictly
fewer unbounded $\until{}{}{}$ operators than $\theta$.  As $\theta$
does not contain unbounded $\since{}{}{}$ operators, $\mathit{pr}(\theta)$ is
finite, so choose $N>\mathit{pr}(\theta)+1$.  Next we apply the following 
equivalence 
\begin{align*}
\until{\psi}{\varphi}{} \quad \longleftrightarrow \quad & \until{\psi}{\varphi}{<N}\\ 
\vee \big(\fBox_{< N} \varphi \wedge \fDia_{=N} (\psi \vee (\varphi \wedge \until{\psi}{\varphi}{}))\big) \, .
\end{align*}

Now $\mathit{pr}(\psi \vee (\varphi \wedge \until{\psi}{\varphi}{})) =
\mathit{pr}(\theta)<N-1$, and the subformulas $\until{\psi}{\varphi}{<N}$ and
$\fBox_{< N} \varphi$ have strictly fewer unbounded $\until{}{}{}$
operators than $\theta$.  So by the induction hypothesis the formula on
the right hand side of the above equivalence is equivalent to one that
is syntactically separated, completing the inductive step.  Similarly
$\since{}{}{}$ formulas that do not contain unbounded $\until{}{}{}$
operators are equivalent to syntactically separated formulas.
Applying these observations to Lemma~\ref{lem:separate} gives our main
result, which we repeat here for completeness.

\begin{theoremNo}{\ref{thm:sepThm}}
Every {MTL} formula is equivalent to a Boolean combination of:
\begin{itemize}
\item $\fDia_{=N}  \varphi$ where $\mathit{pr}(\varphi)<N-1$,
\item $\pDia_{=N} \varphi$ where $\mathit{fr}(\varphi)< N-1$, and
\item $\varphi$ where all intervals occurring in the temporal operators are bounded.
\end{itemize}
\end{theoremNo}

\section{Expressive completeness on bounded formulas}\label{sec:bounded}
In this section we show expressive completeness of \textit{MTL} for a fragment
of {\ONEMLO} consisting of \emph{bounded formulas}, i.e., formulas
$\varphi(x)$ that refer only to a bounded interval around $x$.

Given terms $t_2$ and $t_2$, define $\mathrm{Bet}(t_1,t_2)$ to consist
of {\ONEMLO} formulas in which
\begin{enumerate}
\item[(i)] each subformula $\exists z\, \psi$ has the form
$\exists z\,((t_1 \leq z<t_2) \wedge \chi)$,
  i.e., each quantifier is relativized to the half-open interval between
  $t_1$ (inclusive) and $t_2$ (exclusive);
\item[(ii)] in each atomic subformula $P(t)$ the term
  $t$ is a bound occurrence of a variable.
\end{enumerate}

Clauses (i) and (ii) ensure that a formula in $\mathrm{Bet}(t_1,t_2)$
only refers to the values of monadic predicates on points in the half-open
interval $[t_1,t_2)$.  We say that a formula $\varphi(x)$ in
$\mathrm{Bet}(x-N,x+N)$ is \emph{$N$-bounded} and that $\varphi(x)$ in
$\mathrm{Bet}(x,x+1)$ is a \emph{unit formula}.

Observe that in a unit formula the only essential use of the $+1$
function is in specifying the range of the quantified variables.  More
precisely, we have the following proposition, where $\psi[t/y]$
denotes the formula obtained by substituting term $t$ for all free
occurrences of variable $y$ in $\psi$:

\begin{proposition}
For any unit formula $\varphi(x)$ there is an {\MLO} formula $\psi \in
\mathrm{Bet}(x,y)$ such that $\varphi$ is equivalent to
$\psi[(x+1)/y]$.
\label{prop:unit}
\end{proposition}
\begin{proof}
We show that all uses of the $+1$ function in $\varphi$ other than
to specify the range of quantified variables can be
eliminated.

Let $u,v$ be bound variables and $k_1,k_2 \in \mathbb{N}$.  
Since $u,v$ range over an open interval of length $1$
an
inequality of the form $u + k_1 < v + k_2$ can be replaced by
(i)~$u<v$, if $k_1 = k_2$; (ii)~$\true$, if $k_1 < k_2$; and
(iii)~$\false$ otherwise.  Likewise an equality of the form
$u+k_1 = v+k_2$ can be replaced by $u=v$ if $k_1=k_2$, and $\false$
otherwise.  
\end{proof}

The main result of this section is:
\begin{theorem}\label{thm:bounded}
For every $N$-bounded formula $\varphi(x)$ there exists an equivalent
MTL formula $\varphi^\dag$.
\end{theorem}

In~\cite{ORW09} it was shown that \textit{MTL} is expressively complete for
{\ONEMLO} on bounded domains of the form $[0,N)$.
  Theorem~\ref{thm:bounded} is subtly different from that result,
  which used the definability of the point $0$ in a crucial way.  In
  particular, unlike~\cite{ORW09}, in the present setting we require
  \textit{MTL} operators whose constraining intervals have fractional endpoints
  to achieve expressive completeness.

The proof of Theorem~\ref{thm:bounded} has the following structure:
\begin{enumerate}[Step 1. ]
\item By introducing extra predicates, we rewrite each $N$-bounded
  formula as a Boolean combination of unit formulas and atoms.
\item Using a normal form of Gabbay, Pnueli, Shelah, and
  Stavi~\cite{GPSS80} (see also Hodkinson~\cite{Hod}) we give a
  translation of unit formulas to \textit{MTL}.  This step
  reveals a connection between the granularity of \textit{MTL} and the
  quantifier depth of the unit formulas.
\item We complete the translation by removing the new predicate
  symbols introduced in Step 1.
\end{enumerate}

\subsection*{Step 1. Translation to unit formulas and atoms}
We translate an $N$-bounded formula $\varphi(x)$ into a formula
$\overline{\varphi}(x)$ that is a Boolean combination of unit formulas
and atoms.

Let $\varphi(x)$ mention monadic predicates $P_1,\ldots,P_m$.  For
each predicate $P_i$ we introduce an indexed family of new predicates
$P_i^j$, where $-N \leq j < N$.  Intuitively, $P_i^j(y)$ stands for
$P_i(y+j)$.  Formally, given a signal $f$ that interprets the
$P_i$ we define a signal $\overline{f}$ that interprets the
$P_i^j$ by
\[ P_i^j \in \overline{f}(r) \Longleftrightarrow P_i \in f(r+j)\]
for all $r \in \domain$.

Next we define a formula $\overline{\varphi}$ such that $f,r
\models \varphi$ if and only if $\overline{f},r \models
\overline{\varphi}$.  To obtain $\overline{\varphi}$ we recursively
replace every instance of a subformula
\[\exists y\,((x-N \leq
y < x+N) \wedge \psi)\] in $\varphi$ by the formula
\[ \exists y\, \big((x \leq y < x+1) \wedge
( \psi[(y-N)/y] \vee \ldots \vee \psi[(y+(N-1))/y]) \big).
\]

Having carried out these substitutions, we use simple arithmetic to
rewrite every term in $\varphi$ as $z+k$, where $z$ is a variable and
$k \in \mathbb{Z}$ is an integer constant.
Every use of monadic predicates in $\varphi$ now has the form $P_i(z +
k)$, for $-N \leq k < N$.  Replace every such predicate by $P^k_i(z)$.

After the above operations the resulting formula is a Boolean
combination of unit formulas and atomic formulas.

\subsection*{Step 2. Translating unit formulas to MTL}
In the next stage of the proof we show how to translate unit formulas
into equivalent \textit{MTL} formulas.  Critical to this step is the following
definition and lemma from~\cite{GPSS80}.  Lemma~\ref{lem:gpss} is the
main technical lemma in the expressive completeness proof of \textit{LTL} for
FO($<$) in~\cite{GPSS80}.

A \emph{decomposition formula} $\delta(x,y)$ is any formula of the
form
\begin{align*}
x<y & \wedge \exists z_0 \ldots \exists z_n\,(x=z_0< \cdots < z_n=y) \\
    & \wedge \bigwedge\{ \varphi_i(z_i) : 0 \leq i < n\}\\
    & \wedge \bigwedge\{ \forall u\,((z_{i-1}<u<z_{i}) \rightarrow \psi_i(u)) :
0 < i \leq n\}
\end{align*}
where $\varphi_i$ and $\psi_i$ are \textit{LTL} formulas regarded as unary
predicates.

\begin{lemma}[\cite{GPSS80}]
Over any domain with a complete linear order, every {\MLO} formula $\psi(x,y)$ in
$\mathrm{Bet}(x,y)$ is equivalent to a Boolean combination of
decomposition formulas $\delta(x,y)$.
\label{lem:gpss}
\end{lemma}

Recall from Proposition~\ref{prop:unit} that for any unit formula
$\theta(x)$ there exists an $\mathit{FO}(<)$ formula $\psi \in
\mathrm{Bet}(x,y)$ such that $\psi[(x+1)/y]$ is equivalent to
$\theta(x)$.  Thus, in light of Lemma~\ref{lem:gpss}, to translate
unit formulas to \textit{MTL} it suffices to consider unit formulas of the form
$\delta[(x+1)/y]$ where $\delta(x,y)$ is a a decomposition formula.

\begin{proposition}
Let $\delta(x,y)$ be a decomposition formula and consider the unit
formula $\theta(x) = \delta[(x+1)/y]$.  Then there is an \textit{MTL} formula
equivalent to $\theta(x)$.
\end{proposition}
\begin{proof}
We proceed by induction
on the number $n$ of existential quantifiers in $\delta(x,y)$.

\subsubsection*{Base case}
Let $\delta(x,y) = \varphi(x) \wedge \forall u\, (x<u<y \rightarrow
\psi(u))$, where $\varphi$ and $\psi$ are \textit{LTL} formulas.  Clearly the \textit{MTL} formula
$\varphi \wedge \fBox_{(0,1)} \psi$ is equivalent to $\delta[(x+1)/y]$.

\subsubsection*{Inductive case}
Let $\delta(x,y)$ have the form
 \begin{align*}
x<y & \wedge \exists z_0 \ldots \exists z_n\,(x=z_0< \cdots < z_n=y) \\
    & \wedge \bigwedge\{ \varphi_i(z_i) : 0 \leq i < n\}\\
    & \wedge \bigwedge\{ \forall u\,((z_{i-1}<u<z_i) \rightarrow \psi_i(u)) :
0 < i \leq n\} \, .
\end{align*}

Consider the unit formula $\theta(x):=\delta[(x+1)/y]$.  The idea is
to define \textit{MTL} formulas $\alpha_k,\beta_k$, $0\leq k<2n$, whose
disjunction is equivalent to $\theta$.  The definition of these
formulas is based on a case analysis of the values of the
existentially quantified variables $z_1,\ldots,z_{n-1}$ in $\delta$,
similar to the idea of Example~\ref{ex:fomtl}.
To this end, consider the following $2n$ half-open subintervals of
$[x,x+1)$: $[x,x+\frac{1}{2n}),[x+\frac{1}{2n},x+\frac{2}{2n}),\ldots,
[x+\frac{2n-1}{2n},x+1)$.  We identify three mutually exclusive cases
according to the distribution of the $z_i$ among these intervals:
\begin{enumerate}
\item $\{z_1,\ldots,z_{n-1}\} \subseteq [x+\frac{k}{2n},x+\frac{k+1}{2n})$ for
  some $k<n$;
\item $\{z_1,\ldots,z_{n-1}\} \subseteq [x+\frac{k}{2n},x+\frac{k+1}{2n})$ for some $k$, $n \leq k < 2n$;
\item There exists $k$, $1\leq k<2n$, and $l$, $1\leq l < n-1$, such
  that $z_l < x+\frac{k}{2n} \leq z_{l+1}$ (i.e., $z_1,\ldots,z_{n-1}$
  are not all contained in a single interval).
\end{enumerate}

\paragraph{Case 1.} 
Assume that $k<n$ and consider the following \textit{MTL} formula:
\[
\renewcommand{\arraystretch}{1.5}
\begin{array}{llll}
\alpha_k \, :=&\varphi_0 \wedge \psi_1 \mathrel{\mathbf{U}_{[\frac{k}{2n},\frac{k+1}{2n})}}  
                 \\&\quad( \varphi_1 \wedge 
               ( \psi_2\mathrel{\mathbf{U}_{(0,\frac{1}{2n})}}\\
              &\quad\quad( \varphi_2 \wedge 
               ( \psi_3 \mathrel{\mathbf{U}_{(0,\frac{1}{2n})}}\\
              & \quad\qquad\qquad \ddots\\
              &\quad\qquad( \varphi_{n-2} \wedge
               ( \psi_{n-1} \mathrel{\mathbf{U}_{(0,\frac{1}{2n})}}\\
              &\quad\qquad\quad( \varphi_{n-1} \wedge\fBox_{(0,\frac{1}{2n})} \psi_n)) \cdots )
              \\&\wedge \quad\fBox_{(\frac{k+1}{2n},1)} \psi_n \, .\notag
\end{array}\]
By construction, if $\alpha_k$ holds at a point $x$ then the formulas
$\varphi_0,\psi_1,\varphi_1,\ldots,\varphi_{n-1},\psi_n$ hold in sequence
along the interval $[x,x+1)$.  In particular, $\psi_n$ holds on the
interval starting at the time that the subformula
$\fBox_{(0,\frac{1}{2n})} \psi_n$ begins to hold and extending to time
$x+1$ ( thanks to the ``overlapping'' subformula
$\fBox_{(\frac{k+1}{2n},1)} \psi_n$).  Thus $\alpha_k$ implies
$\theta$.  Conversely, if $\theta$ holds with the existentially
quantified variables $z_1,\ldots,z_{n-1}$ all lying in the interval
$(x+\frac{k}{2n},x+\frac{k+1}{2n})$, then clearly $\alpha_k$ also
holds.
\end{proof}

\paragraph{Case 2.} 

Suppose that $n\leq k<2n$ and consider the following \textit{MTL} formula:
\[
\renewcommand{\arraystretch}{1.5}
\begin{array}{lll}
\alpha_k \, := & \fDia_{=1}  \big[ \psi_n \mathrel{\mathbf{S}_{(\frac{2n-k-1}{2n},\frac{2n-k}{2n})}}\\
&\qquad ( \varphi_{n-1} \wedge  ( \psi_{n-1} \mathrel{\mathbf{S}_{(0,\frac{1}{2n})}}\\
&\qquad\quad( \varphi_{n-2} \wedge ( \psi_{n-2} \mathrel{\mathbf{S}_{(0,\frac{1}{2n})}}\\
&\qquad\qquad\qquad\quad         \ddots \\
&\qquad\qquad\quad( \varphi_{2} \wedge  ( \psi_{2} \mathrel{\mathbf{S}_{(0,\frac{1}{2n})}} \\
&\qquad\qquad\qquad( \varphi_{1} \wedge \pBox_{(0,\frac{1}{2n})} \psi_1)) \cdots )\big ]\\ 
&\wedge \quad \fBox_{(0,\frac{k}{2n})} \psi_1\quad\wedge\quad\varphi_0\, .
\end{array}\]

The definition of $\alpha_k$  is according to similar principles as in
Case 1.  If it holds at a point $x$ then the sequence of past
operators ensures that the formulas
$\psi_n,\varphi_{n-1},\psi_{n-1},\ldots,\varphi_1,\psi_1,\varphi_0$ hold in
sequence, backward from $x+1$ to $x$.  Thus $\alpha_k$ implies
$\theta$.  Conversely, if $\theta$ holds with the existentially
quantified variables $z_1,\ldots,z_{n-1}$ all lying in the interval
$[x+\frac{k}{2n},x+\frac{k+1}{2n})$, $n\leq k<2n$, then clearly
$\alpha_k$ also holds.

\paragraph{Case 3.}  
Suppose that $z_l < x+\frac{k}{2n} \leq z_{l+1}$ for some $k$, $1 \leq
k<2n$, and $l$, $1 \leq l < n-1$.  

The idea is, for each choice of $l$, to decompose $\theta$ into a property $\sigma_l$ holding
on the interval $[x,x+\frac{k}{2n})$ and a property $\tau_l$ holding
on the interval $[x+\frac{k}{2n},x+1)$.  We then apply the induction
hypothesis to transform $\sigma_l$ and $\tau_l$ to equivalent \textit{MTL}
formulas.  To this end, define
\begin{align*}
\sigma_l(x) \, := \, & \exists z_0 \ldots \exists z_{l+1}
                  (x=z_0 < \cdots < z_{l+1} = x+\textstyle\frac{k}{2n})\\
     & \wedge \bigwedge \{ \varphi_i(z_i) : 0 \leq i \leq l \}\\
     & \wedge \bigwedge \{ \forall u((z_{i-1}<u<z_i) \rightarrow \psi_i(u)) : 
                           1 \leq i \leq l+1 \} 
\end{align*}
and
\begin{align*}
\tau_l(x) \, := \, & \exists z_l \ldots \exists z_n
                  (x=z_l < \cdots < z_n = x+\textstyle\frac{2n-k}{2n})\\
     & \wedge \bigwedge \{ \varphi_i(z_i) : l+1 \leq i < n \}\\
     & \wedge \bigwedge \{ \forall u((z_{i-1}<u<z_i) \rightarrow \psi_i(u)) : 
                           l < i \leq n \} \, .
\end{align*}

We can turn $\sigma_l$ into an equivalent \textit{MTL} formula $\sigma_l^\ast$
by the following sequence of transformations: scale by
$\frac{2n}{k}$ to obtain a unit formula, apply the induction
hypothesis to transform the unit formula to an equivalent \textit{MTL} formula,
finally scale the resulting \textit{MTL} formula by $\frac{k}{2n}$.  We
likewise transform $\tau_l$ into an equivalent \textit{MTL} formula
$\tau_l^\ast$. 

We now define
\[ \beta_k := \bigvee_{1 \leq l < n-1} \Big( \sigma_l^\ast \wedge \fDia_{=\frac{k}{2n}} \big((\psi_{l+1}   \wedge \tau_l^\ast) \vee (\varphi_{l+1} \wedge \tau_{l+1}^\ast)\big)\Big) \, .\]
From the definition of $\sigma_l$ it is clear that $\beta_k$ matches
$\theta$ on $[x,x+\frac{k}{2n})$.  For the remaining interval
  $[x+\frac{k}{2n},x+1)$ we distinguish between two cases: if
    $x+\frac{k}{2n} < z_{l+1}$, then $\fDia_{=\frac{k}{2n}}(\psi_{l+1}
    \wedge \tau_l^\ast)$ agrees with $\theta$; and if $x+\frac{k}{2n}
    = z_{l+1}$ then $\fDia_{=\frac{k}{2n}}(\varphi_{l+1} \wedge
    \tau_{l+1}^\ast)$ agrees with $\theta$.  Thus $\beta_k$ implies
    $\theta$.  Conversely if $\theta$ holds with the existentially
    variables $z_1,\ldots,z_{n-1}$ satisfying the conditions of Case 3
    then one of the disjuncts, and hence $\beta_k$, must hold.

\subsection*{Step 3. Completing the translation}

After Step 2 we have an \textit{MTL} formula equivalent to the formula
$\overline{\varphi}(x)$ obtained in Step 1.  It remains only to
eliminate the extra predicates introduced in Step 1.  To this end, for
each predicate $P$ and $j\geq 0$, replace $P^j$ by $\fDia_{=j}P$, and
for $j<0$ replace $P^j$ by $\pDia_{=-j}P$.  Finally we obtain an \textit{MTL}
formula $\varphi^\dag$ equivalent to the original $N$-bounded formula
$\varphi(x)$.

\begin{theoremNo}{\ref{thm:bounded}}
For every $N$-bounded {\ONEMLO} formula $\varphi(x)$ there exists
an equivalent MTL formula $\varphi^\dag$.
\end{theoremNo}

\section{Expressive completeness of MTL}\label{sec:mtl}
Our next step towards proving the expressive completeness of
\textit{MTL} is to show that it is able to express all of {\ONEMLO}.
\begin{lemma}
For every {\ONEMLO} formula $\varphi(x)$ there is an equivalent MTL formula
$\varphi^\dag$.
\label{lem:main}
\end{lemma}
\begin{proof}
The proof is by induction on the quantifier depth $n$ of $\varphi$.

\subsubsection*{Base case, $n=0$}  
All atoms are of the form $P_i(x)$, $x=x$, $x<x$, $x+1=x$.  We replace
these by $P_i$, $\true$, $\false$, $\false$ respectively and obtain an
\textit{MTL} formula which is clearly equivalent to $\varphi$.

\subsubsection*{Inductive case}  
Without loss of generality we may assume $\varphi = \exists
y. \psi(x,y)$, where $\psi(x,y)$ has quantifier depth $n-1$.  We would
like to remove $x$ from $\psi$.  To this end we take a disjunction
over all possible choices for $\gamma:\{P_1(x), \ldots
P_m(x)\}\to\{\true,\false\}$, and use $\gamma$ to determine the value
of $P_i(x)$ in each disjunct via the formula $\theta_\gamma :=
\bigwedge_{i=1}^m (P_i(x) \leftrightarrow \gamma(P_i))$.  Thus we can
equivalently write $\varphi$ in the form
\begin{gather}
\bigvee_\gamma \big(\theta_\gamma(x) \wedge 
 \exists y. \psi_\gamma(x,y)\big) \, ,
\label{eq:pull}
\end{gather}
where the propositions $P_i(x)$ do not appear in the $\psi_\gamma$.

Now in each $\psi_\gamma$, $x$ appears only in atoms of the form
$x=z$, $x<z$, $x>z$, $x+1=z$, $x=z+1$ for some variable $z$.  We next
introduce new monadic propositions $P_=$, $P_<$, $P_>$, $P_+$ and
$P_-$, and replace each of the atoms containing $x$ in $\psi_\gamma$
with the corresponding proposition.  That is, $x=z$ becomes $P_=(z)$,
$x<z$ becomes $P_<(z)$ and so on.  This yields a formula
$\psi'_\gamma(y)$ in which $x$ does not occur, such that
$\psi'_\gamma(y)$ has the same truth value as $\psi_\gamma(x,y)$ if
the interpretations of the new propositions are consistent with $x$.
Thus for each value of $x$, (\ref{eq:pull}) has the same truth value as 
\begin{gather} 
\bigvee_\gamma (\theta_\gamma(x) \wedge \exists y.\psi'_\gamma(y)) \, .
\label{eq:pull2}
\end{gather}
for suitable intepretations of the new propositions.

By the induction hypothesis, for each $\gamma$ there is an \textit{MTL}
formula $\theta_\gamma^\dag$ equivalent to $\theta_\gamma(x)$, and an
\textit{MTL} formula $\psi^\dag_\gamma$ equivalent to
$\psi'_\gamma(y)$.  Then our
original formula $\varphi$ has the same truth value at each point $x$ as
\[ \varphi' := \bigvee_\gamma \big(\theta_\gamma^\dag \wedge 
(\pDia \psi_\gamma^{\dagger} \vee \psi_\gamma^{\dagger} \vee \fDia 
\psi_\gamma^{\dagger})\big)\]
for suitable interpretations of $\{ P_=, P_<, P_>, P_+,P_-\}$.  

By Theorem~\ref{thm:sepThm}, $\varphi'$ is equivalent to a Boolean
combination of formulas
\begin{enumerate}[(I) ]
\item $\fDia_{=N}  \theta$ where $\mathit{pr}(\theta)<N-1$,
\item $\pDia_{=N} \theta$ where $\mathit{fr}(\theta)< N-1$, and
\item $\theta$ where all intervals occurring in the temporal operators are bounded.
\end{enumerate}
Now in formulas of type (I) above, we know the intended value of each
of the propositional variables $P_=, P_<, P_>, P_+,P_-$: they are all
$\false$ except $P_<$, which is $\true$.  So we can replace these
propositional atoms by $\true$ and $\false$ as appropriate and obtain
an equivalent \textit{MTL} formula which does not mention the new
variables.  Likewise we know the value of each of propositional
variables in formulas of type (II): all are $\false$ except $P_>$,
which is $\true$; so we can again obtain an equivalent \textit{MTL}
formula which does not mention the new variables.  It remains to deal
with each of the bounded formulas, $\theta$.  From
Proposition~\ref{prop:easy}, there exists a formula $\theta^\ast(x)$
in {\QMLO}, with predicates from $\{P_=, P_<, P_>, P_+,P_-\}$, which
is equivalent to $\theta$.  It is not difficult to see that as
$\theta$ is bounded, there is an $N$ such that $\theta^\ast$ is
$N$-bounded.  We now unsubstitute each of the introduced propositional
variables.  That is, replace in $\theta^\ast(x)$ all occurrences of
$P_=(z)$ with $z=x$, all occurrences of $P_<(z)$ with $x<z$ etc.  The
result is an equivalent formula $\theta^+ \in$ {\QMLO}, which is still
$N$-bounded as we have not removed any constraints on the variables of
$\theta^\ast$.  From Theorem~\ref{thm:bounded}, it follows that there
exists an \textit{MTL} formula $\delta$ that is equivalent to
$\theta^+$, i.e., equivalent to $\theta$.
\end{proof}

Finally, recall from Section~\ref{sec:exp-equiv} how a translation
from {\ONEMLO} to {\textit{MTL}} can be lifted to a translation {\QMLO} to
{\textit{MTL}} via a simple scaling argument.  Thus Lemma~\ref{lem:main}
entails our main result:

\begin{theoremNo}{\ref{thm:main}}
For every {\QMLO} formula $\varphi(x)$ there is an equivalent MTL formula
$\varphi^\dag$.
\end{theoremNo}

\section{Conclusion}

In general, the theory of real-time verification lacks the stability
and canonicity of the classical theory, and has tended to suffer from
a proliferation of competing and mismatching formalisms.  Thus it was
a pleasant surprise to discover that \textit{MTL} is expressively complete for
first-order logic, particularly in view of the extensive literature on
the former and the fact that the latter is a natural yardstick against
which to measure expressiveness.

We are currently investigating the full extent of this result,
including a version for \textit{MTL} with integer constants, equipped with
counting modalities.

%








\end{document}